\documentclass[letterpaper]{article} % DO NOT CHANGE THIS
\usepackage[preprint]{aaai2027}
\usepackage[hyphens]{url}  % DO NOT CHANGE THIS
\usepackage{graphicx} % DO NOT CHANGE THIS
\usepackage{natbib}  % DO NOT CHANGE THIS AND DO NOT ADD ANY OPTIONS TO IT
\usepackage{caption} % DO NOT CHANGE THIS AND DO NOT ADD ANY OPTIONS TO IT
\usepackage{algorithm}
\usepackage{algorithmic}
\usepackage{amsmath,amssymb,amsthm,mathtools,bm}
\usepackage{booktabs}

\newtheorem{proposition}{Proposition}
\newtheorem{theorem}{Theorem}
\newtheorem{assumption}{Assumption}

\graphicspath{{Figures/}}
\title{GPU-Parallelization of Markov Chain Pool Decoding with Unbiased MCMC}
\author{
Takato Ueno\textsuperscript{\rm 1},
Shuji Kijima\textsuperscript{\rm 2}
}
\affiliations{
\textsuperscript{\rm 1}Graduate School of Data Science, Shiga University\\
\textsuperscript{\rm 2}Faculty of Data Science, Shiga University
}

\begin{document}
\maketitle

\begin{abstract}
Markov chain pool decoding (MCPD) devised by Knill et al. (1996) identifies likely positive clones from noisy pooled-test results. The standard MCPD estimates clone-wise posterior probabilities using Gibbs sampling, but it may allocate excessive computational effort to low-scoring clones. This paper focuses on parallelizing MCPD on GPU architectures. Whereas the standard MCPD employs systematic-scan updates, we propose a score-weighted update scheme that updates high-scoring clones more frequently. We prove that the stationary distribution of the proposed Markov chain coincides with the target posterior distribution. To enable efficient GPU parallelization, we further incorporate the unbiased MCMC framework of Jacob et al. (2020) and employ a slot-refilling technique based on the arguments by Glynn and Heidelberger (1991) about the coupling of Markov chains. Experiments involving 1,298 clones, 97 pools, and three true positives demonstrate improved recovery compared with uniform decoders, while maintaining high overlap under high-noise conditions.\end{abstract}

\section{Introduction}

In a pooling experiment, a large number of clones are assigned to a small number of pools, and positive clones are identified from the observed pool results. Because of pool-level false negatives and false positives, the set of positive clones that is consistent with the observations is not uniquely determined. \citet{BrunoKnillBaldingEtAl1995} constructed random $k$-sets designs and $k$-sets packing designs and applied them to a human chromosome 16 YAC library consisting of 1,298 clones. \citet{KnillSchliepTorney1996} introduced a Bayesian model that accounts for observation errors and estimated the posterior positive probability of each clone by Markov chain pool decoding. These posterior probabilities are used to rank the clones that are to be subjected to a confirmatory assay. Since the Bayesian decoder explicitly conditions on false-positive and false-negative rates, and since pooled testing error can vary across assay and pooling conditions, the finite-budget behavior of MCPD estimators should be examined across observation-noise regimes \citep{KnillSchliepTorney1996,QianRefsniderMooreEtAl2020,TanOmarLeeWong2020,TatsuokaChenLu2023}.\citet[p.~399]{KnillSchliepTorney1996} accumulated statistics over 10,000 steps after 1,000 warmup steps and determined the run length empirically by comparing several runs.

\citet{KnillSchliepTorney1996} estimates clone-wise posterior probabilities but may allocate excessive computation to low-scoring clones, although positive clones are few in library screening.

In this work, we increase the update frequency of high-score clones using score weights fixed from the observations before the chain starts. We then employ unbiased MCMC; the unbiased MCMC method proposed by \citet{JacobOLearyAtchade2020} removes initialization bias by computing a bias-correction term based on a telescoping sum constructed from coupled chains while ordinary MCMC suppresses the influence of the initial state by using a warmup period to bring the chain close to the target distribution, and the required run length increases when mixing is slow\citep{BiswasJacobVanetti2019,AtchadeJacob2024}. We also employ the framework of \citet{glynn1991analysis} for GPU parallelization. This work makes the following three contributions.

\begin{itemize}
\item We propose a parallel computation method for MCPD of \citet{KnillSchliepTorney1996} from posterior-mean estimation by a long MCMC chain into GPU-parallel aggregation of short-chain unbiased estimators.
\item We construct a random-scan Gibbs sampler with score weights that are fixed from the observed data before the chain starts, and we show that its transition matrix leaves the target posterior distribution invariant.
\item We combine common-random-number coupling with a completion-time correction and extend unbiased estimation of the clone-wise posterior positive probability to fixed-budget GPU execution with slot refill.
\end{itemize}

\paragraph{Other Related Work}

In the study of pooling designs, \citet{BarillotLacroixCohen1991} analyzed multidimensional pooling strategies theoretically, and \citet{BaldingTorney1996} treated design conditions that take error detection into account. \citet{BaldingBrunoKnillTorney1996} compared non-adaptive pooling designs. These studies concern the construction and comparison of pooling designs, whereas the present work decodes observations obtained from a fixed design and changes the MCMC estimator used for MCPD. Regarding random-scan Gibbs samplers, \citet{LevineCasella2006} proposed a method that updates the selection probabilities on the basis of past samples, and \citet{LatuszynskiRobertsRosenthal2013} gave convergence conditions for adaptive Gibbs samplers together with examples of non-convergence. In contrast, the score weights used here are fixed from the observed pool outcomes before the chain starts, and the contribution is to prove invariance of the resulting MCPD transition matrix rather than to analyze adaptive scan rules.

\citet{GonzalezLowGrettonGuestrin2011} constructed simultaneous updates of conditionally independent variables by graph coloring and parallel updates of junction-tree blocks. \citet{TereninDongDraper2019} executed the Gibbs update of an exchangeable latent-variable model in a data-parallel manner on a GPU.

\section{Preliminaries}

\subsection{Knill's Method}

Let $z_i\in\{0,1\}$ denote the state of clone $i$, let $y_j\in\{0,1\}$ denote the observation for pool $j$, and let $A\in\{0,1\}^{m\times n}$ be the pooling matrix. The latent state of a pool is defined by $\xi_j(\bm z)=\mathbf 1\{\sum_{i=1}^n A_{ji}z_i>0\}$. With prior positive probability $q$, false-positive probability $\alpha_{\mathrm{fp}}$, and false-negative probability $\beta_{\mathrm{fn}}$, the posterior distribution provided by \citet{KnillSchliepTorney1996} is
\[
\pi(\bm z\mid \bm y)\propto
\prod_{i=1}^n q^{z_i}(1-q)^{1-z_i}
\prod_{j=1}^m \Pr\{y_j\mid \xi_j(\bm z)\}.
\]
The quantity to be estimated for clone $i$ is $p_i=\mathbb E_\pi[Z_i]$.

Let $\mathcal P(i)$ be the set of pools that contain clone $i$. \citet[p.~399]{KnillSchliepTorney1996} defined the naive ranking as
\[
R_i
=\prod_{j\in\mathcal P(i)}
\frac{\Pr\{y_j\mid \xi_j=1\}}
{\Pr\{y_j\mid \xi_j=0\}}.
\]
This ranking is used for pruning: low-ranked clones are fixed to negative, and only the top clones are included in the state of the Markov chain. Let $\mathcal A$ be the set of retained clones and $n_a$ be its cardinality. Fixing the clones in $\mathcal A^c$ to $0$, the posterior distribution after pruning is given by
\[
\begin{aligned}
\pi_{\mathcal A}(\bm z_{\mathcal A}\mid \bm y)
&=\Pr\{\bm Z_{\mathcal A}=\bm z_{\mathcal A}
\mid \bm y,\bm Z_{\mathcal A^c}=\bm 0\}\\
&\propto
\prod_{i\in\mathcal A}q^{z_i}(1-q)^{1-z_i}
\prod_{j=1}^m
\Pr\{y_j\mid \xi_j^{\mathcal A}(\bm z_{\mathcal A})\},
\end{aligned}
\]
where $\xi_j^{\mathcal A}(\bm z_{\mathcal A})=\mathbf 1\{\sum_{i\in\mathcal A}A_{ji}z_i>0\}$. This is the renormalization of the posterior distribution on condition that $\bm Z_{\mathcal A^c}=\bm 0$. This work takes this pruned posterior as the estimation target. In what follows we relabel the clones in $\mathcal A$ as $1,\ldots,n_a$ and abbreviate $\pi_{\mathcal A}$ as $\pi$.

Let the state space be $\Omega=\{0,1\}^{n_a}$. Write the full conditional of clone $i$ as $p_i(b\mid \bm x_{-i})=\Pr_\pi\{Z_i=b\mid \bm Z_{-i}=\bm x_{-i}\}$ for $b\in\{0,1\}$, and define the single-site update map by
\[
[\Phi_i(\bm x,u)]_r
=
\begin{cases}
\mathbf 1\{u\le p_i(1\mid \bm x_{-i})\}, & r=i,\\
x_r, & r\ne i.
\end{cases}
\]
With $U\sim\operatorname{Unif}(0,1)$, we let $P_i(\bm x,\bm x')=\Pr\{\Phi_i(\bm x,U)=\bm x'\}$ be the single-site Gibbs transition matrix. One sweep of the systematic scan is $Q=P_1P_2\cdots P_{n_a}$.

For random updates, we count one single-site update as one Markov transition and use
\[
P_{\mathrm{unif}}=\frac1{n_a}\sum_{i=1}^{n_a}P_i,
\qquad
P_w=\sum_{i=1}^{n_a} w_iP_i .
\]
The weights satisfy $w_i>0$ and $\sum_iw_i=1$. Because $I_t$ is generated at each step from one and the same distribution, independently of the current state and of the past, we have, for every $t$,
\[
\Pr\{\bm X_{t+1}=\bm x'\mid \bm X_t=\bm x\}
=\sum_{i=1}^{n_a}\Pr\{I_t=i\}P_i(\bm x,\bm x'),
\]
and the right-hand side does not depend on the time $t$. The right-hand side equals $P_{\mathrm{unif}}$ for the uniform random update and $P_w$ for the score-weighted update. \citet{LiuWongKong1995,LiuWongKong1994} distinguished systematic scan from random scan and treated random scan as a mixture of single-site transition matrices. The MCPD of \citet[p.~398]{KnillSchliepTorney1996} uses systematic scan, which updates all clones in a fixed order. We compare this implementation with random update, score-weighted update, coupled random update, and coupled score-weighted update.

The quantity $\widehat\pi_{b,T}(h)=\frac{1}{T}\sum_{t=b}^{b+T-1}h(\bm X_t)$ is the ordinary finite MCMC average, and its expectation $\mathbb E[\widehat\pi_{b,T}(h)]=\frac{1}{T}\sum_{t=b}^{b+T-1}\nu P^t(h)$ depends on the initial distribution $\nu$. This dependence is the source of the initialization bias of a finite-length average.The random-update Gibbs sampler is shown in Algorithm~\ref{alg:random-update}.

\subsection{Unbiased MCMC}

Unbiased MCMC couples, through common random numbers, two Markov chains whose marginal transitions are governed by the same transition matrix, and adds the differences before meeting to a finite-length average as a telescoping correction. Under a faithful coupling that maintains $\bm X_t=\bm Y_{t-1}$ after meeting, $H_{k:\ell}(h)$ is an unbiased estimator of $\pi(h)$. \citet{JacobOLearyAtchade2020} denote the corresponding construction by $H_{k:m}$ and call it a time-averaged estimator. We write $\ell$ in place of their $m$ and call $k:\ell$ the averaging window. \citet{JacobOLearyAtchade2020} established the marginal convergence, the moment bound, the meeting-time tail, and the post-meeting agreement that this construction requires. \citet{BiswasJacobVanetti2019} and \citet{AtchadeJacob2024} organize the construction and the diagnostics of couplings.

The computation of a replication terminates once both the averaging-window end $\ell$ and the meeting time $\tau_1$ have been reached, so the completion cost is random. \citet{WangBlanchetGlynn2024} compared unbiased and biased estimators in terms of total computation and completion time, and analyzed conditions under which the former shortens the completion time in a massively parallel environment.

\section{Proposed Method}

\subsection{Score-Weighted Update Gibbs Sampler}

Let the log score of the naive ranking used for pruning be $\eta_i=\log\{q/(1-q)\}+\log R_i$. The first term is common to all clones, so $\eta_i$ and $R_i$ induce the same ranking.

Let $d_j^{\mathcal A}=\sum_{i\in\mathcal A}A_{ji}$ be the size of pool $j$ restricted to the active clones, and define
\[
a_j
=\log\Pr\{y_j\mid \xi_j=1\}
-\log\Pr\{y_j\mid \xi_j=0\},
\]
together with $s_i=\sum_{j:A_{ji}=1}d_j^{\mathcal A}a_j$. The quantity $s_i$ is a fixed score computed only once, before the chain starts, from the observed pool outcomes and the pool sizes after pruning. We set $\gamma=100n_a$, $s_{\max}=\max_i s_i$, and $s_{\min}=\min_i s_i$; we let $\tau_w=\log\gamma/(s_{\max}-s_{\min})$ when $s_{\max}>s_{\min}$ and $\tau_w=0$ when $s_{\max}=s_{\min}$; and we use
\[
w_i=(1-\alpha_w)
\frac{\exp(\tau_ws_i)}{\sum_{j=1}^{n_a}\exp(\tau_ws_j)}
+\frac{\alpha_w}{n_a},
\qquad \alpha_w=0.5 .
\]
This corresponds to $\epsilon=\{\alpha_w/(1-\alpha_w)\}\sum_j\exp(\tau_ws_j)/n_a$. The uniform component makes the weight of every site positive. The Score-weighted update Gibbs sampler is shown in Algorithm~\ref{alg:score-weighted-update}.

\begin{algorithm}[t]
\caption{Random update Gibbs sampler}
\label{alg:random-update}
\begin{algorithmic}[1]
\REQUIRE Initial distribution $\nu$, burn-in $b$, sampling length $T$, clone index $j$
\ENSURE Finite-length average $\widehat\pi_{b,T}^{\mathrm{RU}}(h_j)$ for clone $j$
\STATE Draw $\bm X_0\sim\nu$.
\FOR{$t=0,\ldots,b+T-2$}
\STATE Draw $I_t\sim\operatorname{Unif}\{1,\ldots,n_a\}$ and $U_t\sim\operatorname{Unif}(0,1)$.
\STATE Set $\bm X_{t+1}=\Phi_{I_t}(\bm X_t,U_t)$.
\ENDFOR
\STATE Set $h_j(\bm z)=z_j$ and return $\widehat\pi_{b,T}^{\mathrm{RU}}(h_j)=\frac{1}{T}\sum_{t=b}^{b+T-1}h_j(\bm X_t)$.
\end{algorithmic}
\end{algorithm}

\begin{algorithm}[t]
\caption{Score-weighted update Gibbs sampler}
\label{alg:score-weighted-update}
\begin{algorithmic}[1]
\REQUIRE Initial distribution $\nu$, burn-in $b$, sampling length $T$, fixed weights $w$, clone index $j$
\ENSURE Finite-length average $\widehat\pi_{b,T}^{\mathrm{SW}}(h_j)$ for clone $j$
\STATE Draw $\bm X_0\sim\nu$.
\FOR{$t=0,\ldots,b+T-2$}
\STATE Draw $I_t$ with $\Pr\{I_t=i\}=w_i$ and draw $U_t\sim\operatorname{Unif}(0,1)$.
\STATE Set $\bm X_{t+1}=\Phi_{I_t}(\bm X_t,U_t)$.
\ENDFOR
\STATE Set $h_j(\bm z)=z_j$ and return $\widehat\pi_{b,T}^{\mathrm{SW}}(h_j)=\frac{1}{T}\sum_{t=b}^{b+T-1}h_j(\bm X_t)$.
\end{algorithmic}
\end{algorithm}

\subsection{Coupled Score-Weighted Update Gibbs Sampler}

We construct an unbiased estimator from the score-weighted update Gibbs sampler and a coupling. One single-site update is counted as one Markov transition. We draw $\bm Z_0\sim\nu$ and set $\bm X_0=\bm Y_0=\bm Z_0$. First, only $\bm X$ is updated once, which produces $\bm X_1$. At each later time $t\ge1$, a site $I_t$ and a uniform variable $U_t$ shared by the two chains are generated, and both chains are updated once by
\[
\bm X_{t+1}=\Phi_{I_t}(\bm X_t,U_t),
\qquad
\bm Y_t=\Phi_{I_t}(\bm Y_{t-1},U_t),
\]
where $\Pr\{I_t=i\}=w_i$ and $U_t\sim\operatorname{Unif}(0,1)$. The marginal transition matrix of each chain coincides with $P_w$. The meeting time is defined by $\tau_1=\inf\{t\ge1:\bm X_t=\bm Y_{t-1}\}$.  The Coupled Score-weighted update Gibbs sampler is shown in Algorithm~\ref{alg:coupled-score-weighted}.

\begin{algorithm}[t]
\caption{Coupled score-weighted update Gibbs sampler}
\label{alg:coupled-score-weighted}
\begin{algorithmic}[1]
\REQUIRE Initial distribution $\nu$, fixed weights $w$, averaging-window end $\ell$
\ENSURE Coupled paths $\{\bm X_t\}$, $\{\bm Y_t\}$, meeting time $\tau_1$
\STATE Draw $\bm Z_0\sim\nu$ and set $\bm X_0=\bm Y_0=\bm Z_0$.
\STATE Draw $I_0$ with $\Pr\{I_0=i\}=w_i$ and draw $U_0\sim\mathrm{Unif}(0,1)$.
\STATE Set $\bm X_1=\Phi_{I_0}(\bm X_0,U_0)$.
\STATE Set $t\leftarrow1$ and record whether $\bm X_t=\bm Y_{t-1}$.
\WHILE{$t<\ell$ or $\bm X_t\ne\bm Y_{t-1}$}
\STATE Draw $I_t$ with $\Pr\{I_t=i\}=w_i$ and draw $U_t\sim\mathrm{Unif}(0,1)$.
\STATE Set $\bm X_{t+1}=\Phi_{I_t}(\bm X_t,U_t)$ and $\bm Y_t=\Phi_{I_t}(\bm Y_{t-1},U_t)$.
\STATE After meeting, keep using the same $I_t,U_t$ so that $\bm X_{t+1}=\bm Y_t$ is preserved.
\STATE Set $t\leftarrow t+1$ and record the first $t$ at which $\bm X_t=\bm Y_{t-1}$.
\ENDWHILE
\STATE Return the coupled paths and $\tau_1$.
\end{algorithmic}
\end{algorithm}

\subsection{Coupled Estimator and Fixed-Budget GPU Execution}

With $h_j(\bm z)=z_j$, define
\[
H_r(h_j)
=h_j(\bm X_r)+\sum_{s=1}^{\infty}
\{h_j(\bm X_{r+s})-h_j(\bm Y_{r+s-1})\}
\]
and $H_{k:\ell}(h_j)=\frac{1}{\ell-k+1}\sum_{r=k}^{\ell}H_r(h_j)$, where $k:\ell$ is the averaging window.

The GPU scheduler maps a logical slot $p$ to one persistent CUDA thread. Each slot holds the two chain states, the random-number state, the elapsed progress units, the number of completions, and the local sums of the estimator. In the execution treated by Theorem~\ref{thm:ct}, replication $r$ is advanced until both the averaging-window end and the meeting have been reached. Writing $\tau_{p,r}$ for the meeting time, the progress-unit completion cost of the one-lag estimator is $C_{p,r}=\max\{\ell,\tau_{p,r}\}$. At that point $\widehat\theta_{p,r}$ is finalized, and the slot updates its local sums and its number of completions. If $S_p(r)<B$, the state is initialized and the next replication starts. We call this operation slot refill.

\citet[Section~3, Equation~(3.2), Proposition~3.2]{glynn1991analysis} aggregate, with equal weights across processors, the local averages of the replications completed by each processor up to a fixed horizon. In the present work the processor index, the replication output, the runtime, and the time horizon correspond to $p$, $\widehat\theta_{p,r}$, $C_{p,r}$, and $B$, respectively. Here $B$ is the progress-unit horizon up to which the start of a new replication is permitted.

Fix a clone $j$ and set $\widehat\theta_{p,r}=H_{k:\ell}^{(p,r)}(h_j)$. For replication $r$ of logical slot $p$, define
\[
\begin{gathered}
S_p(0)=0,
\qquad
S_p(r)=\sum_{u=1}^rC_{p,u},\\
N_p(B)=\max\{r\ge0:S_p(r)\le B\}.
\end{gathered}
\]
Here $N_p(B)$ is the number of completions within the horizon. Only when $N_p(B)=0$ do we retain, as a zero-completion correction, the first replication that completes beyond $B$. Accordingly, we define the slot mean $A_p(B)$ by
\[
\begin{gathered}
N_p^*(B)=\max\{1,N_p(B)\},\\
A_p(B)
=\frac1{N_p^*(B)}
\sum_{r=1}^{N_p^*(B)}\widehat\theta_{p,r},
\end{gathered}
\]
and define $\widehat\theta_{\mathrm{CT}}(h_j)=\frac{1}{P}\sum_{p=1}^PA_p(B)$. The aggregation method is shown in Algorithm~\ref{alg:unbiased-aggregation}.

\begin{algorithm}[t]
\caption{Fixed-budget GPU execution and completion-time-corrected aggregation}
\label{alg:unbiased-aggregation}
\begin{algorithmic}[1]
\REQUIRE $P$ logical slots fixed before the outcomes are observed, per-slot horizon $B$, averaging window $k:\ell$, clone index $j$
\ENSURE Aggregated estimator $\widehat\theta_{\mathrm{CT}}(h_j)$
\FOR{each logical slot $p$}
\STATE Run i.i.d.\ coupled pairs sequentially with Algorithm~\ref{alg:coupled-score-weighted}.
\STATE Compute $\widehat\theta_{p,r}=H_{k:\ell}^{(p,r)}(h_j)$ from each completed pair.
\STATE After pair $r$ completes, refill the same slot with the next pair if $S_p(r)<B$.
\IF{at least one pair completes by $B$}
\STATE Retain all replications with $S_p(r)\le B$ and discard the second and later replications that complete beyond $B$.
\ELSE
\STATE Run the first replication through to meeting and retain $H_{k:\ell}^{(p,1)}(h_j)$.
\ENDIF
\STATE Compute the slot mean $A_p(B)$.
\ENDFOR
\STATE Return $\widehat\theta_{\mathrm{CT}}(h_j)=\frac{1}{P}\sum_{p=1}^PA_p(B)$.
\end{algorithmic}
\end{algorithm}

\section{Theory}

We establish, in this order, the validity of the fixed-weight transition matrix, the unbiasedness of the single-pair coupled estimator, and the unbiasedness of the fixed-budget parallel aggregation. The proofs are given in the supplementary material.

\begin{assumption}[Finite positive posterior and fixed weights]
\label{ass:finite}
The state space is finite and $\pi(\bm z)>0$ for every state. The weights are fixed before the chain starts, satisfy $w_i>0$ and $\sum_iw_i=1$, and depend neither on the state nor on the time.
\end{assumption}

\citet[p.~398]{KnillSchliepTorney1996} also assumed that the pool outcome likelihood takes neither the value $0$ nor the value $1$, and derived the irreducibility and aperiodicity of MCPD from this condition. In the binary observation model used here, $0<q<1$, $0<\alpha_{\mathrm{fp}}<1$, and $0<\beta_{\mathrm{fn}}<1$ guarantee $\pi(\bm z)>0$ at every state after pruning. The condition on fixed and positive weights is the condition that the present work adds for the score-weighted random scan.

\begin{proposition}[Fixed-weight random-scan Gibbs transition matrix]
\label{prop:transition-matrix}
Under Assumption~\ref{ass:finite}, $\pi$ is the invariant distribution of $P_w$, and $P_w$ is irreducible and aperiodic. Consequently, $\nu P_w^t\to\pi$ for every initial distribution $\nu$.
\end{proposition}

\begin{theorem}[Unbiasedness of the coupled estimator, and finiteness of its variance and expected computation time]
\label{thm:cmcpd}
Under Assumption~\ref{ass:finite}, generate a coupled pair with Algorithm~\ref{alg:coupled-score-weighted} and run it through to meeting. For every integer $0\le k\le\ell$ and every clone $j$, the following hold.
\begin{enumerate}
\item[(a)] Each marginal transition matrix is $P_w$, and $\mathbb E[h_j(\bm X_t)]\to\pi(h_j)$.
\item[(b)] For every $\eta>0$, $\sup_{t\ge0}\mathbb E[|h_j(\bm X_t)|^{2+\eta}]\le1$.
\item[(c)] There exist $c<\infty$ and $\rho\in(0,1)$ such that $\Pr(\tau_1>t)\le c\rho^t$.
\item[(d)] $\bm X_t=\bm Y_{t-1}$ holds almost surely for every $t\ge\tau_1$.
\end{enumerate}
Consequently, $\mathbb E[H_{k:\ell}(h_j)]=\pi(h_j)=\Pr_\pi(Z_j=1)$, and the variance and the expected computation time of $H_{k:\ell}(h_j)$ are finite.
\end{theorem}

The fixed-completion-time estimator of \citet[Equation~(3.2), Proposition~3.2]{glynn1991analysis} preserves the single-replication expectation when the outputs and the runtimes are jointly i.i.d. In the present work we apply this result to the single-pair estimator $\widehat\theta_{p,r}$.

\begin{assumption}[Parallel replication]
\label{ass:parallel}
The observed data, the set of clones to be analyzed, the initial distribution $\nu$, the fixed weights $w$, and the averaging window $k:\ell$ are fixed. The logical slots $1,\ldots,P$ are fixed before the outcomes are observed. Each pair $(p,r)$ is assigned an i.i.d.\ ideal random stream, and each replication is run through to meeting with the same initial distribution and the same transition kernel. Then $(\widehat\theta_{p,r},C_{p,r})$ is i.i.d.\ within a slot and across slots. Dependence between the estimator and the computation time within the same pair is allowed.
\end{assumption}

Assumption~\ref{ass:parallel} corresponds to the jointly i.i.d.\ setting for outputs and runtimes of \citet[Section~2]{glynn1991analysis}.

\begin{theorem}[Equal-weight logical-slot aggregate]
\label{thm:ct}
Under Assumption~\ref{ass:finite}, Assumption~\ref{ass:parallel}, and Algorithm~\ref{alg:unbiased-aggregation}, we have $\mathbb E[\widehat\theta_{\mathrm{CT}}(h_j)]=\pi(h_j)$ for every clone $j$.
\end{theorem}

Theorem~\ref{thm:ct} is a direct application of \citet[Equation~(3.2), Proposition~3.2]{glynn1991analysis} to the single-pair estimator obtained in Theorem~\ref{thm:cmcpd}. The supplementary material gives the correspondence between the assumptions and the notation. The implementation computes $\widehat{\bm\theta}_{p,r}=(H_{k:\ell}^{(p,r)}(h_1),\ldots,H_{k:\ell}^{(p,r)}(h_{n_a}))^\mathsf{T}$ from a common coupled pair and aggregates it as
\[
\widehat{\bm p}_{\mathrm{CT}}
=
\frac1P\sum_{p=1}^P
\frac1{N_p^*(B)}
\sum_{r=1}^{N_p^*(B)}
\widehat{\bm\theta}_{p,r}.
\]
Because all components use the common $C_{p,r}$ and $N_p(B)$, applying the theorem to each component gives $\mathbb E[\widehat{\bm p}_{\mathrm{CT}}]=(\pi(h_1),\ldots,\pi(h_{n_a}))^\mathsf{T}$.

The global pooled mean, which collects the completed replications of all slots directly, is
\[
\widehat\theta_{\mathrm{pool}}
=
\frac{\sum_{p=1}^{P}\sum_{r=1}^{N_p(B)}
\widehat\theta_{p,r}}
{\sum_{p=1}^{P}N_p(B)} .
\]
This aggregation uses runtime-dependent slot weights that are proportional to the completion counts. The expectation-preserving identity of Theorem~\ref{thm:ct} holds for $\widehat\theta_{\mathrm{CT}}(h_j)$, which aggregates the within-slot averages with equal weights. Our implementation uses $\widehat\theta_{\mathrm{CT}}(h_j)$.

\section{Experiments}

\subsection{Experimental Setup and Metrics}

\subsubsection{Synthetic data.}
For the synthetic-data experiments we used a Knill-type random pooling design with 1,298 clones, 97 pools, and 3 true positives. Each clone belongs to 3 pools. Following \citet{KnillSchliepTorney1996}, we set $q=2.6/1298=0.002003$, so the unpruned prior expected number of positive clones was $1298q=2.6$. After pruning by the naive score, 333 clones were analyzed. Setting the false-negative rate and the false-positive rate to either 0.05 or 0.10 produced four noise conditions. The observed pool data were fixed for each noise condition, and each cell was repeated with 10 chain seeds.

Varying FN and FP while holding the pooling design and latent positive clones fixed isolates the effect of observation ambiguity from changes in the underlying screening instance. These conditions evaluate whether the relative accuracy of the update schemes, including the coupling correction, depends on the noise level. The experiment uses the same FN/FP values for data generation and posterior evaluation; it therefore evaluates finite-budget approximation of correctly specified posteriors rather than robustness to likelihood misspecification. Fixing the observed pool data within each noise condition further makes the variation across executions attributable to the Monte Carlo procedures and chain seeds.

We compared five implementations: systematic scan, random update, coupled random update, score-weighted update, and coupled score-weighted update.

\subsubsection{Tapestry real data.}
As public real data we used Tapestry 320 and Tapestry 961 reported by \citet{ChakravarthyEtAl2020Tapestry}. Tapestry 320 consists of 320 samples, 48 pools, and 5 positives; each sample corresponds to 3 pools and each pool to 20 samples. Tapestry 961 consists of 961 samples, 93 pools, and 10 positives; each sample corresponds to 3 pools and each pool to 31 samples. For Tapestry 320 and Tapestry 961, we set $q=5/320=0.015625$ and $q=10/961=0.010406$, respectively, so the unpruned prior means matched the reported positive counts. We retained the same clone-level prior probabilities after pruning.

The published pooling matrices and observed pool outcomes reduce the discretion involved in choosing a pooling design; the FN/FP likelihood conversion and pruning remain analysis choices. We converted these matrices and outcomes into the binary MCPD likelihood. For Tapestry 320 we set FN and FP to 0.01. For Tapestry 961, one of the 23 true-positive pools was not detected, so we set FN to $1/23=0.0435$ and FP to 0.01. For Tapestry 320 we retained all 320 clones, and for Tapestry 961 we retained the top 247 clones by naive score.

Let $n_a$ denote the number of analyzed clones and $s\in\{1,3,5,10,20\}$ the total estimator budget in sweeps. We set $M_s=s n_a$ single-site updates and $b_s=\operatorname{round}(0.1M_s)$. The uncoupled estimators discarded the first $b_s$ updates and averaged the following $M_s-b_s$ updates. The coupled estimators used $k=b_s+1$ and $\ell=M_s$, giving $\ell-k+1=M_s-b_s$ terms in the base average. The per-slot progress-unit horizon was fixed at 20 sweeps in every cell. The uncoupled methods used 2,000,000 logical chain slots and processed approximately $N_s=2{,}000{,}000\times20/s$ chains in batches. To match the number of chain states held simultaneously, the coupled methods used 1,000,000 logical slots and advanced one coupled pair in each slot. When a pair completed within the 20-sweep horizon, slot refill started another pair, and partial pairs were excluded. In a slot with zero completions, the first pair continued through to meeting, and the slot means were aggregated with equal weights. The 20-sweep meeting cap makes the reported coupled estimates finite-cap approximations to the uncapped unbiased estimators in Theorems~\ref{thm:cmcpd} and \ref{thm:ct}. For the reference vector we used systematic scan with 2,000,000 chains, 1,000 warmup sweeps, and 10,000 sampling sweeps.

\subsubsection{Evaluation Metrics}

Let $\mathcal A$ be the set of analyzed clones. The all-clone MAE is evaluated by
\[
\operatorname{MAE}_{\mathcal A}
=\frac1{|\mathcal A|}
\sum_{i\in\mathcal A}
|\widehat p_i-p_i^{\mathrm{ref}}| .
\]
For ranking accuracy, we compared the 10 clones with the largest reference posterior probabilities and the 10 clones selected by each method. The top-10 overlap is the fraction of the 10 reference top-ranked clones that also appear in the method-selected top-10 set. The top-10 MAE is the average absolute difference between the estimated and reference posterior probabilities over the 10 reference top-ranked clones. Thus, top-10 overlap measures candidate-set agreement, whereas top-10 MAE measures the error of the posterior magnitude assigned to the top clones of the reference.

The estimator-length figures report means and standard errors over 10 executions. The focused top-10 table and the Tapestry table report the 1-sweep mean and standard error over 10 executions. The wall-clock table first averages the 5 estimator lengths within each execution and then reports the mean and standard error over 10 executions.

Tables and plots covering all evaluated noise conditions, estimator lengths, and real-data datasets are provided in the supplementary material.

\subsection{Focused Top-10 Comparison under FN/FP Noise}
Table~\ref{tab:focused-top10} reports top-10 MAE and top-10 overlap for the five methods with 1 sweep under the symmetric noise settings FN $=$ FP $=0.05$ and FN $=$ FP $=0.10$. The results show that coupling correction reduces top-10 MAE and that score-weighted update improves top-10 MAE accuracy relative to random update.

In the case of FN $=$ FP $= 0.10$, coupled score-weighted update attained the smallest top-10 MAE, $6.427 \times 10^{-4}$, followed by coupled random update at $1.318 \times 10^{-3}$, whereas random update showed the largest top-10 MAE, $3.853 \times 10^{-2}$. Both coupled methods attained smaller top-10 MAE than the three uncoupled methods. However, we note that random update attained the largest top-10 overlap, $0.990$.

In both cases, FN $=$ FP $= 0.05$ and FN $=$ FP $= 0.10$, random update attained the largest top-10 overlap and the largest top-10 MAE. The large top-10 MAE indicates that accurate identification of top-ranked clones does not guarantee accurate posterior magnitudes. Coupled score-weighted update reduced the top-10 MAE more than score-weighted update, with a 94.0\% improvement at FN $=$ FP $= 0.10$ compared with a 78.2\% improvement at FN $=$ FP $= 0.05$. This difference supports the benefit of correcting initialization bias from a finite averaging window in the stronger-noise instance evaluated here.

\begin{table}[t]
\centering
\caption{FN/FP focused top-10 comparison at 1 sweep.}
\label{tab:focused-top10}
\scriptsize
\setlength{\tabcolsep}{1pt}
\begin{tabular}{lcc}
\toprule
Method & Top-10 MAE & Top-10 overlap \\
\midrule
\multicolumn{3}{c}{FN $=0.05$, FP $=0.05$} \\
\midrule
systematic scan & $(2.419{\pm}0.000){\times}10^{-2}$ & $0.400{\pm}0.000$ \\
random update & $(1.272{\pm}0.000){\times}10^{-1}$ & $\mathbf{0.900{\pm}0.000}$ \\
coupled random update & $(3.449{\pm}0.485){\times}10^{-2}$ & $0.820{\pm}0.020$ \\
score-weighted update & $(1.236{\pm}0.000){\times}10^{-2}$ & $0.700{\pm}0.000$ \\
coupled score-weighted update & $\mathbf{(2.695{\pm}0.240){\times}10^{-3}}$ & $0.850{\pm}0.027$ \\
\midrule
\multicolumn{3}{c}{FN $=0.10$, FP $=0.10$} \\
\midrule
systematic scan & $(4.624{\pm}0.002){\times}10^{-3}$ & $0.400{\pm}0.000$ \\
random update & $(3.853{\pm}0.000){\times}10^{-2}$ & $\mathbf{0.990{\pm}0.010}$ \\
coupled random update & $(1.318{\pm}0.208){\times}10^{-3}$ & $0.780{\pm}0.020$ \\
score-weighted update & $(1.071{\pm}0.000){\times}10^{-2}$ & $0.200{\pm}0.000$ \\
coupled score-weighted update & $\mathbf{(6.427{\pm}0.711){\times}10^{-4}}$ & $0.890{\pm}0.018$ \\
\bottomrule
\end{tabular}

\end{table}

\subsection{MAE by Estimator Length at the Maximum Noise Condition}

Next, Figure~\ref{fig:max-noise-mae} plots the all-clone MAE and top-10 MAE of the five methods after 1, 3, 5, 10, and 20 sweeps when FN $=$ FP $= 0.10$. The coupled methods show smaller MAE when the estimator length is short, whereas the MAE is larger when the estimator length is long; this comes from that uncoupled averages retain more initialization bias.

The upper panel refers to the 333 clones after pruning and the lower panel to the reference top-10 clones; the vertical axis is logarithmic. For all-clone MAE, coupled score-weighted update was the most accurate of the five methods at 1 and 3 sweeps, while systematic scan was the most accurate at 5, 10, and 20 sweeps. Coupled score-weighted update attained the smallest top-10 MAE at 1, 3, and 5 sweeps, with values of $6.427\times10^{-4}$, $7.451\times10^{-4}$, and $5.757\times10^{-4}$, respectively. Systematic scan attained the smallest values at 10 and 20 sweeps.

The advantage of coupled score-weighted update persisted through 5 sweeps for top-10 MAE and through 3 sweeps for all-clone MAE, supporting the claimed benefit at short estimator lengths. At longer estimator lengths, systematic scan approached the long-run systematic-scan reference. Longer averages leave less initialization bias for coupling correction to remove, which reduces the relative advantage of the coupled methods.

\begin{figure*}[t]
\centering
\includegraphics[width=0.86\linewidth]{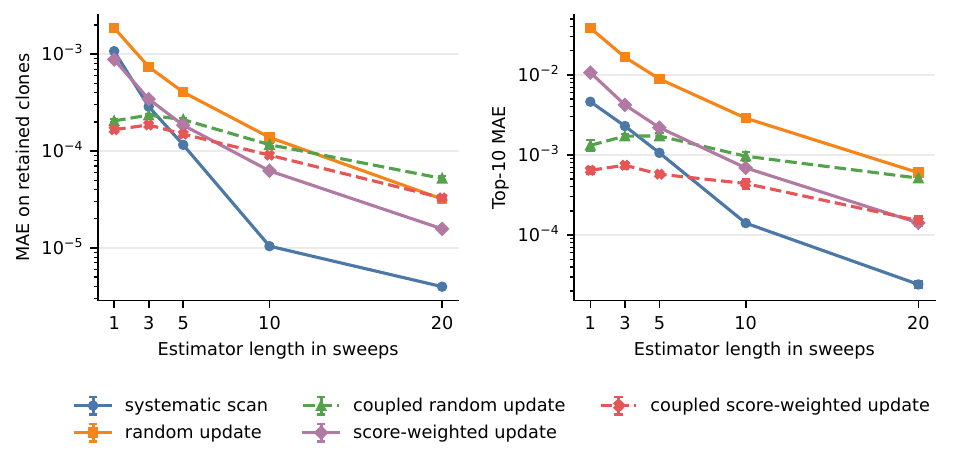}
\caption{MAE by estimator length at FN $=$ FP $=0.10$.}
\label{fig:max-noise-mae}
\end{figure*}

\subsection{Execution Characteristics at the Maximum Noise Condition}

\subsubsection{Slot refill ablation.}
Figure~\ref{fig:refill-ablation} compares coupled score-weighted update with and without slot refill. The vertical axis shows the all-clone MAE for the 333 clones after pruning on a logarithmic scale. At an estimator length of 1 sweep, slot refill reduced the all-clone MAE by 70.7\%, from $5.675\times10^{-4}$ to $1.662\times10^{-4}$, and produced lower MAE than no refill at lengths from 1 to 10 sweeps. Refill executed additional pairs in slots that became free within the 20-sweep horizon and increased the number of aggregated estimators. At 20 sweeps, no progress units remained for starting another pair.

\begin{figure}[t]
\centering
\includegraphics[width=0.78\linewidth]{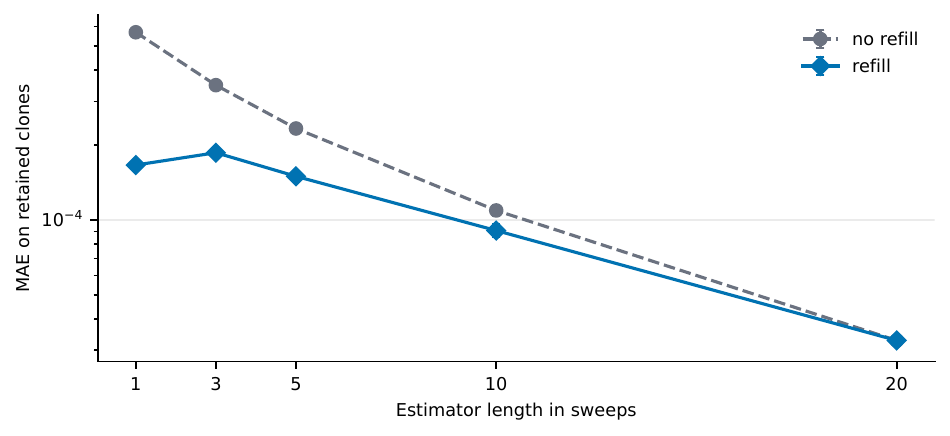}
\caption{Slot refill ablation.}
\label{fig:refill-ablation}
\end{figure}

\subsubsection{Meeting time.}
Table~\ref{tab:meeting-time-summary} shows the mean of the maximum meeting time for coupled random update and coupled score-weighted update. The mean maximum meeting time decreased from 14.156 sweeps to 12.917 sweeps.

\begin{table}[t]
\centering
\caption{Meeting-time summary.}
\label{tab:meeting-time-summary}
\scriptsize
\setlength{\tabcolsep}{2pt}
\begin{tabular}{lr}
\toprule
Method & Mean maximum \\
\midrule
coupled random update & 14.156 \\
coupled score-weighted update & 12.917 \\
\bottomrule
\end{tabular}

\end{table}

\subsubsection{Wall-clock.}
Table~\ref{tab:wall-clock-main} shows the GPU wall-clock at FN 0.10 and FP 0.10. The reference systematic scan is the long-chain execution used to define $p_i^{\mathrm{ref}}$. Coupled score-weighted update was the fastest evaluated method, running 4.8\% faster than systematic scan and 14.9\% faster than score-weighted update.

\begin{table}[t]
\centering
\caption{GPU wall-clock.}
\label{tab:wall-clock-main}
\scriptsize
\begin{tabular}{lcc}
\toprule
Method & Mean wall-clock & Standard error \\
\midrule
reference systematic scan & 836.8894 & 0.0000 \\
systematic scan & 2.7454 & 0.0008 \\
random update & 2.7509 & 0.0003 \\
score-weighted update & 3.0711 & 0.0001 \\
\textbf{coupled score-weighted update} & \textbf{2.6131} & \textbf{0.0003} \\
\bottomrule
\end{tabular}

\end{table}

\subsection{Top-10 Comparison on Tapestry Real Data}

The Tapestry experiments test two claims about performance on public pooling designs. Score weighting should improve the coupled estimator across designs, while the effect of coupling on the score-weighted estimator may vary between datasets. Table~\ref{tab:tapestry-real-data} reports the 1-sweep top-10 MAE and top-10 overlap against the long-run systematic-scan reference.

On Tapestry 320, coupled score-weighted update obtained the smallest top-10 MAE among the five methods, 0.0373. This value was 78.1\% smaller than that of score-weighted update and 90.1\% smaller than that of coupled random update. Random update and score-weighted update obtained the largest top-10 overlap, 0.900. On Tapestry 961, coupled score-weighted update improved the top-10 overlap of coupled random update from 0.550 to 0.690 and reduced its top-10 MAE by 36.9\%, from 0.7321 to 0.4623. However, score-weighted update attained the best values on both metrics, with a top-10 overlap of 0.800 and a top-10 MAE of 0.2305.

\begin{table}[t]
\centering
\caption{Tapestry top-10 comparison at 1 sweep.}
\label{tab:tapestry-real-data}
\scriptsize
\setlength{\tabcolsep}{2pt}
\begin{tabular}{lcc}
\toprule
Method & Top-10 MAE & Top-10 overlap \\
\midrule
\multicolumn{3}{c}{Tapestry 320} \\
\midrule
systematic scan & \(0.4576{\pm}0.0000\) & \(0.300{\pm}0.000\) \\
random update & \(0.4163{\pm}0.0000\) & \(\mathbf{0.900{\pm}0.000}\) \\
coupled random update & \(0.3767{\pm}0.0522\) & \(0.690{\pm}0.038\) \\
score-weighted update & \(0.1700{\pm}0.0000\) & \(\mathbf{0.900{\pm}0.000}\) \\
coupled score-weighted update & \(\mathbf{0.0373{\pm}0.0068}\) & \(0.870{\pm}0.021\) \\
\midrule
\multicolumn{3}{c}{Tapestry 961} \\
\midrule
systematic scan & \(0.6539{\pm}0.0000\) & \(0.100{\pm}0.000\) \\
random update & \(0.5152{\pm}0.0000\) & \(0.700{\pm}0.000\) \\
coupled random update & \(0.7321{\pm}0.0648\) & \(0.550{\pm}0.060\) \\
score-weighted update & \(\mathbf{0.2305{\pm}0.0000}\) & \(\mathbf{0.800{\pm}0.000}\) \\
coupled score-weighted update & \(0.4623{\pm}0.0308\) & \(0.690{\pm}0.060\) \\
\bottomrule
\end{tabular}

\end{table}

\section{Conclusion}
We propose a parallel computation method for MCPD as fixed-budget GPU aggregation of short-chain unbiased estimators. The proposed score-weighted random scan uses weights computed once from the observed pool outcomes and restricted pool sizes. Because each single-site Gibbs kernel $P_i$ leaves the pruned posterior invariant, so does $P_w=\sum_i w_iP_i$. We coupled two $P_w$ chains through common site indices and uniform variates and applied the time-averaged estimator $H_{k:\ell}$ to remove initialization bias. Persistent logical slots refill completed pairs within a fixed horizon; equal weighting of slot means preserves the single-pair expectation under jointly i.i.d.\ outputs and runtimes.

With 1,298 clones, 97 pools, and three true positives, coupled score-weighted update attained the smallest top-10 MAE at short estimator lengths under the strongest synthetic noise. Tapestry confirmed the benefit of score weighting across both designs, although score-weighted update outperformed its coupled counterpart on Tapestry 961. Further improvements will require an MCPD-specific coupling and careful tuning of the lag and burn-in.

\section{Acknowledgments}
This work is partially supported by JSPS KAKENHI Grant Numbers\\
JP23K21645 and JP26K02866.

\clearpage
\raggedbottom
\bibliography{references}

\clearpage
\appendix
\section*{Supplementary Material}
\setcounter{theorem}{0}
\setcounter{proposition}{0}
\setcounter{assumption}{0}
\setcounter{table}{0}
\setcounter{figure}{0}
\renewcommand{\thetheorem}{S\arabic{theorem}}
\renewcommand{\theproposition}{S\arabic{proposition}}
\renewcommand{\theassumption}{S\arabic{assumption}}
\renewcommand{\thetable}{S\arabic{table}}
\renewcommand{\thefigure}{S\arabic{figure}}
\section{Purpose and Structure of the Supplementary Material}

This supplement presents, in order, the derivations and arguments for the fixed-weight coupled estimator and the random-completion-aware GPU aggregation that were omitted from the main text. The assumption, the proposition, and the two theorems restate those of the main text and are numbered in the same order.

\section{Fixed-Weight Random-Scan Gibbs Transition Matrix}

Let the state space be $\Omega=\{0,1\}^n$ and write $\pi$ for the pool-decoding posterior distribution. A transition matrix $P$ on $\Omega$ satisfies
\[
P(\bm x,\bm y)\ge0,
\qquad
\sum_{\bm y\in\Omega}P(\bm x,\bm y)=1 .
\]
Representing a probability distribution $\mu$ as a row vector, the distribution after one transition is $\mu P$. If $\pi P=\pi$, then $\pi$ is the invariant distribution of $P$.

Write the full conditional of coordinate $i$ as $p_i(b\mid\bm x_{-i})=\Pr_\pi\{Z_i=b\mid\bm Z_{-i}=\bm x_{-i}\}$. Let $P_i$ be the single-site Gibbs transition matrix that updates coordinate $i$ from this full conditional and leaves the other coordinates fixed, that is,
\[
P_i(\bm x,\bm y)
=
\mathbf 1\{\bm x_{-i}=\bm y_{-i}\}
p_i(y_i\mid\bm x_{-i}).
\]
For fixed weights $w_1,\ldots,w_n$ we define $P_w=\sum_{i=1}^nw_iP_i$.

\begin{assumption}[Finite positive posterior and fixed weights]
\label{supp:ass:finite}
The state space is finite and $\pi(\bm x)>0$ for every $\bm x\in\Omega$. The weights are fixed before the chain starts, satisfy
\[
w_i>0,
\qquad
\sum_{i=1}^nw_i=1,
\]
and depend neither on the state nor on the time.
\end{assumption}

Assumption~\ref{supp:ass:finite} implies $\pi(\bm x_{-i})>0$ for every $\bm x\in\Omega$ and every $i$, so the full conditionals are well defined. Moreover,
\[
\sum_{\bm y\in\Omega}P_i(\bm x,\bm y)
=\sum_{b\in\{0,1\}}p_i(b\mid\bm x_{-i})
=1,
\]
\[
\sum_{\bm y\in\Omega}P_w(\bm x,\bm y)
=\sum_{i=1}^nw_i
\sum_{\bm y\in\Omega}P_i(\bm x,\bm y)
=1,
\]
so both $P_i$ and $P_w$ are transition matrices.

Proposition~\ref{supp:prop:transition-matrix} states that the marginal transition matrix of the fixed-weight random scan leaves the posterior distribution invariant and converges to it from every initial distribution.

\begin{proposition}[Fixed-weight random-scan Gibbs transition matrix]
\label{supp:prop:transition-matrix}
Under Assumption~\ref{supp:ass:finite}, $\pi$ is the invariant distribution of $P_w$, and $P_w$ is irreducible and aperiodic. Consequently, for every initial distribution $\nu$,
\[
\nu P_w^t\longrightarrow\pi .
\]
\end{proposition}

\begin{proof}
Fix a coordinate $i$ and a state $\bm y\in\Omega$. The only states $\bm x$ that contribute to the sum $(\pi P_i)(\bm y)$ are those with $\bm x_{-i}=\bm y_{-i}$. Such states can be written as $(\bm y_{-i},b)$ with $b\in\{0,1\}$. Hence
\begin{align*}
(\pi P_i)(\bm y)
&=\sum_{\bm x\in\Omega}\pi(\bm x)P_i(\bm x,\bm y)\\
&=\sum_{b\in\{0,1\}}
\pi(\bm y_{-i},b)p_i(y_i\mid\bm y_{-i})\\
&=p_i(y_i\mid\bm y_{-i})
\sum_{b\in\{0,1\}}\pi(\bm y_{-i},b)\\
&=p_i(y_i\mid\bm y_{-i})\pi(\bm y_{-i})\\
&=\pi(\bm y).
\end{align*}
Therefore $\pi P_i=\pi$ for every $i$. By the linearity of the matrix product,
\begin{align*}
\pi P_w
&=\pi\left(\sum_{i=1}^nw_iP_i\right)
=\sum_{i=1}^nw_i(\pi P_i)\\
&=\sum_{i=1}^nw_i\pi
=\pi\sum_{i=1}^nw_i
=\pi .
\end{align*}

Next we show irreducibility. Fix arbitrary $\bm x,\bm x'\in\Omega$. Consider the path that, in $n$ consecutive updates, selects the coordinates $1,2,\ldots,n$ in this order and updates coordinate $i$ to $x'_i$. The probability $w_i$ of selecting each coordinate is positive. By Assumption~\ref{supp:ass:finite}, both $0$ and $1$ have positive probability under each binary full conditional. Hence this prescribed path has positive probability, and $\bm x'$ is reachable from $\bm x$. Therefore $P_w$ is irreducible.

At any state $\bm x$, the event that one coordinate is selected and its current value is drawn again has positive probability. Therefore $P_w(\bm x,\bm x)>0$. Since a finite irreducible Markov chain with a positive self-transition probability has period one, $P_w$ is aperiodic. The convergence theorem for finite-state Markov chains gives the uniqueness of $\pi$ and total-variation convergence from every initial distribution \citep[Theorem~4.9]{LevinPeres2017}. Concretely,
\[
\left\|\nu P_w^t-\pi\right\|_{\mathrm{TV}}
\le
\sum_{\bm x\in\Omega}\nu(\bm x)
\left\|P_w^t(\bm x,\cdot)-\pi\right\|_{\mathrm{TV}}
\longrightarrow0 .
\]
\end{proof}

\section{Lag Coupling}

Write $\Phi_i(\bm x,u)$ for the result of applying to the state $\bm x$ the inverse-transform Gibbs update that uses coordinate $i$ and the uniform variable $u$. Draw the initial state $\bm Z_0\sim\nu$. Generate an i.i.d.\ sequence $\{(I_t,U_t)\}_{t\ge0}$ independent of $\bm Z_0$ in which, for each $t$, $I_t$ and $U_t$ are independent and
\[
\Pr(I_t=i)=w_i,
\qquad
U_t\sim\operatorname{Unif}(0,1).
\]
Set the initial states to $\bm X_0=\bm Y_0=\bm Z_0$ and define the pre-lag update by $\bm X_1=\Phi_{I_0}(\bm X_0,U_0)$. For each $t\ge1$, update both chains with the same $(I_t,U_t)$ by
\[
\bm X_{t+1}=\Phi_{I_t}(\bm X_t,U_t),
\qquad
\bm Y_t=\Phi_{I_t}(\bm Y_{t-1},U_t).
\]
The meeting time is defined by
\[
\tau_1
=\inf\{t\ge1:\bm X_t=\bm Y_{t-1}\}.
\]

\section{Coupled Estimator and Parallel Aggregation}

For clone $j$, define $h_j(\bm x)=x_j$. Then
\[
\pi(h_j)
=\sum_{\bm x\in\Omega}x_j\pi(\bm x)
=\Pr_\pi(Z_j=1).
\]
For an integer $r\ge0$, define
\[
H_r(h_j)
=h_j(\bm X_r)
+\sum_{s=1}^{\infty}
\{h_j(\bm X_{r+s})-h_j(\bm Y_{r+s-1})\},
\]
and for $0\le k\le\ell$, define
\[
H_{k:\ell}(h_j)
=\frac1{\ell-k+1}
\sum_{r=k}^{\ell}H_r(h_j).
\]
\citet{JacobOLearyAtchade2020} call the corresponding $H_{k:m}$ construction a time-averaged estimator. We write $\ell$ in place of their $m$ and call $k:\ell$ the averaging window.

\section{Unbiasedness of the Coupled Estimator and Finiteness of Its Variance and Expected Computation Time}

Theorem~\ref{supp:thm:cmcpd} verifies the four conditions required for coupling-based unbiased MCMC and applies them directly to the posterior marginal probability of each clone. Condition (a) corresponds to marginal validity, condition (b) to the moment bound, condition (c) to the geometric tail of the meeting time, and condition (d) to faithfulness. Computation time is measured in units of one single-site update.

\begin{theorem}[Unbiasedness of the coupled estimator, and finiteness of its variance and expected computation time]
\label{supp:thm:cmcpd}
Under Assumption~\ref{supp:ass:finite}, run a coupled pair generated by the fixed-weight random scan until the averaging-window end $\ell$ has been reached and the lagged pair has met. For every integer $0\le k\le\ell$ and every clone $j$, the following hold.
\begin{enumerate}
\item[(a)] Each marginal transition matrix is $P_w$, and
\[
\mathbb E[h_j(\bm X_t)]\longrightarrow\pi(h_j).
\]
\item[(b)] For every $\eta>0$,
\[
\sup_{t\ge0}
\mathbb E[|h_j(\bm X_t)|^{2+\eta}]
\le1 .
\]
\item[(c)] There exist $c<\infty$ and $\rho\in(0,1)$ such that
\[
\Pr(\tau_1>t)\le c\rho^t .
\]
\item[(d)] For every $t\ge\tau_1$,
\[
\bm X_t=\bm Y_{t-1}
\qquad\text{almost surely}.
\]
\end{enumerate}
Consequently,
\[
\mathbb E[H_{k:\ell}(h_j)]
=\pi(h_j)
=\Pr_\pi(Z_j=1),
\]
and the variance and the expected computation time of $H_{k:\ell}(h_j)$ are finite.
\end{theorem}

\begin{proof}
\par\smallskip\noindent\textbf{Condition (a): marginal validity.}
Let $\bm x,\bm y$ be the current states of the lagged pair. Conditionally on $I_t=i$, inverse-transform sampling makes the marginal transition of the $\bm X$ chain $P_i(\bm x,\cdot)$ and that of the $\bm Y$ chain $P_i(\bm y,\cdot)$. Averaging over coordinates, the marginal transition of the $\bm X$ chain is
\[
\sum_{i=1}^nw_iP_i(\bm x,\cdot)
=P_w(\bm x,\cdot),
\]
and the same computation gives $P_w(\bm y,\cdot)$ for the $\bm Y$ chain.

Since $\bm X_0$ and $\bm Y_0$ have the same distribution and the marginal transitions of both chains are $P_w$, induction gives
\[
\mathcal L(\bm X_t)
=\mathcal L(\bm Y_t)
=\nu P_w^t
\qquad
\text{for every }t\ge0 .
\]
Setting $A_j=\{\bm x\in\Omega:x_j=1\}$, we have $h_j=\mathbf 1_{A_j}$. Because Proposition~\ref{supp:prop:transition-matrix} gives $\nu P_w^t\to\pi$ in total variation,
\begin{align*}
\left|
\mathbb E[h_j(\bm X_t)]-\pi(h_j)
\right|
&=
\left|
\bigl(\nu P_w^t\bigr)(A_j)-\pi(A_j)
\right|\\
&\le
\lVert\nu P_w^t-\pi\rVert_{\mathrm{TV}}
\longrightarrow0 .
\end{align*}
This proves condition (a).

\par\smallskip\noindent\textbf{Condition (b): moment bound.}
For every state $\bm x$ we have $h_j(\bm x)\in\{0,1\}$. Hence, for every $\eta>0$,
\[
|h_j(\bm x)|^{2+\eta}\le1 .
\]
Taking expectations and then the supremum over $t$ gives condition (b).

\par\smallskip\noindent\textbf{Condition (c): geometric tail of the meeting time.}
By Assumption~\ref{supp:ass:finite} and the finiteness of $\Omega$,
\[
\delta
=\min_{\bm x\in\Omega}
\min_{1\le i\le n}
\min\{p_i(1\mid\bm x_{-i}),1-p_i(1\mid\bm x_{-i})\}
>0 .
\]

Suppose the lagged pair just before an update is
\[
(\bm X_t,\bm Y_{t-1})=(\bm x,\bm x'),
\]
and that the common update selects coordinate $i$. If the common uniform variable satisfies
\[
U_t>\max\{p_i(1\mid\bm x_{-i}),p_i(1\mid\bm x'_{-i})\},
\]
then the updated coordinate of both chains becomes $0$. The probability of this event is
\begin{align*}
&1-\max\{p_i(1\mid\bm x_{-i}),p_i(1\mid\bm x'_{-i})\}\\
&\quad=\min\{1-p_i(1\mid\bm x_{-i}),1-p_i(1\mid\bm x'_{-i})\}\\
&\quad\ge\delta .
\end{align*}

Next, take $n$ consecutive coupled updates as one block. Let $E$ be the event that the common coordinate selection picks the coordinates $1,2,\ldots,n$ in this order and that each common uniform variable sets the updated coordinate of both chains to $0$. Within the same block, a coordinate already updated to $0$ is not selected again. Hence, on $E$, both states of the lagged pair equal $0^n$ at the end of the block. Conditionally on the state at the start of the block, the probability of $E$ is uniformly at least
\[
q_{\mathrm{meet}}
=\left(\prod_{i=1}^nw_i\right)\delta^n
>0 .
\]
Here $q_{\mathrm{meet}}$ is a uniform lower bound on the probability of meeting along the particular successful path given by the event $E$. Since other coordinate orders and common updates to $1$ also produce meeting, the within-block meeting probability is at least $q_{\mathrm{meet}}$.

Even conditionally on no meeting in all previous blocks, the probability of meeting in the next block is at least $q_{\mathrm{meet}}$. Repeated use of conditional probabilities together with the Markov property gives
\[
\Pr(\tau_1>1+rn)
\le(1-q_{\mathrm{meet}})^r
\quad
\text{for every integer }r\ge0 .
\]
Set $\rho=(1-q_{\mathrm{meet}})^{1/n}\in(0,1)$. Choosing a finite constant $c$ that covers the at most $n$ time points between block boundaries,
\[
\Pr(\tau_1>t)
\le c\rho^t
\qquad
\text{for every }t\ge0 .
\]

The geometric tail makes all positive moments finite. Using the tail-sum bound for an integer $m\ge1$,
\begin{align*}
\mathbb E[\tau_1^m]
&\le
\sum_{t=0}^{\infty}
\{(t+1)^m-t^m\}
\Pr(\tau_1>t)\\
&\le
c\sum_{t=0}^{\infty}
\{(t+1)^m-t^m\}\rho^t
<\infty .
\end{align*}
Non-integer positive orders are bounded by moments of a larger integer order. This proves condition (c).

\par\smallskip\noindent\textbf{Condition (d): faithfulness.}
Suppose $\bm X_t=\bm Y_{t-1}$. Then the two full conditional probabilities at the selected coordinate coincide. Applying the same uniform variable yields the same Bernoulli outcome. The coordinates that are not selected already agree before the update, so
\[
\bm X_{t+1}=\bm Y_t
\qquad
\text{almost surely}.
\]
Applying this implication inductively from the first meeting time gives
\[
\bm X_t=\bm Y_{t-1}
\quad
\text{for every }t\ge\tau_1
\quad\text{almost surely}.
\]

\par\smallskip\noindent\textbf{Unbiasedness.}
For $R\ge1$, define the quantity obtained by truncating the correction sum at $R$ terms,
\[
H_r^{(R)}(h_j)
=h_j(\bm X_r)
+\sum_{s=1}^{R}
\{h_j(\bm X_{r+s})-h_j(\bm Y_{r+s-1})\}.
\]
By condition (a), the marginal distributions at the same time coincide, so
\[
\mathbb E[h_j(\bm Y_u)]
=\mathbb E[h_j(\bm X_u)]
\qquad
\text{for every }u\ge0 .
\]
Substituting this identity,
\begin{align*}
\mathbb E[H_r^{(R)}(h_j)]
&=\mathbb E[h_j(\bm X_r)]\\
&\quad+
\sum_{s=1}^{R}
\{\mathbb E[h_j(\bm X_{r+s})]
-\mathbb E[h_j(\bm X_{r+s-1})]\}\\
&=\mathbb E[h_j(\bm X_{r+R})],
\end{align*}
where the intermediate expectations cancel between adjacent terms. By condition (a),
\[
\lim_{R\to\infty}
\mathbb E[H_r^{(R)}(h_j)]
=\pi(h_j).
\]

By conditions (c) and (d), almost surely $\tau_1<\infty$ and the correction terms with time index at least $\tau_1$ vanish. Hence
\[
H_r^{(R)}(h_j)\xrightarrow{R\to\infty}H_r(h_j)
\]
almost surely. Since $|h_j|\le1$ and the number of non-zero correction terms is at most $\tau_1$,
\[
|H_r^{(R)}(h_j)|
\le1+2\tau_1
\qquad
\text{for every }R .
\]
Condition (c) gives $\mathbb E[\tau_1]<\infty$, so the right-hand side is integrable. Applying dominated convergence,
\[
\mathbb E[H_r(h_j)]
=\lim_{R\to\infty}
\mathbb E[H_r^{(R)}(h_j)]
=\pi(h_j).
\]
Using the linearity of expectation over the averaging window $k:\ell$,
\begin{align*}
\mathbb E[H_{k:\ell}(h_j)]
&=\frac1{\ell-k+1}
\sum_{r=k}^{\ell}
\mathbb E[H_r(h_j)]\\
&=\pi(h_j).
\end{align*}
Finally,
\[
\pi(h_j)
=\sum_{\bm x\in\Omega}x_j\pi(\bm x)
=\Pr_\pi(Z_j=1).
\]

\par\smallskip\noindent\textbf{Finite variance.}
From the pathwise bound above and condition (c),
\[
\mathbb E[(1+2\tau_1)^2]<\infty .
\]
Hence $\mathbb E[H_r(h_j)^2]<\infty$ for every $r$. By Jensen's inequality for a finite average,
\[
H_{k:\ell}(h_j)^2
\le
\frac1{\ell-k+1}
\sum_{r=k}^{\ell}H_r(h_j)^2 .
\]
Taking expectations on both sides gives
\[
\mathbb E[H_{k:\ell}(h_j)^2]<\infty ,
\]
so the variance is finite.

\par\smallskip\noindent\textbf{Finite expected cost.}
The algorithm performs one update in the pre-lag phase. Each coupled update performs at most two single-site updates. The computation terminates once the averaging-window end $\ell$ has been reached and the lagged pair has met. Hence the total number of single-site updates is at most
\[
1+2\max\{\ell,\tau_1\} .
\]
Since condition (c) gives $\mathbb E[\tau_1]<\infty$, the expected computation time measured in these units is finite.
\end{proof}

\section{Fixed-Budget GPU Execution and Parallel Aggregation}

Consider an integer number $P\ge1$ of logical slots fixed before the outcomes are observed, together with a finite per-slot horizon $B\in(0,\infty)$. Each logical slot corresponds to one persistent CUDA thread. Fix a clone $j$ and set
\[
\widehat\theta_{p,r}
=H_{k:\ell}^{(p,r)}(h_j)
\]
for replication $r$ of slot $p$. A replication continues until both the averaging-window end and the meeting have been reached. Writing $\tau_{p,r}$ for the meeting time, the positive progress-unit completion cost of the one-lag estimator is
\[
C_{p,r}
=\max\{\ell,\tau_{p,r}\}.
\]
Here $B$ is the horizon that selects the replications to be included in the aggregate, and it is distinct from a finite meeting cap that would stop a pair before meeting. In this section no finite meeting cap is imposed, and every replication that is started is run through to meeting. The cumulative computation time and the number of completions up to the horizon $B$ are defined by
\[
S_p(0)=0,
\qquad
S_p(r)
=\sum_{u=1}^{r}C_{p,u},
\]
\[
N_p(B)
=\max\{r\ge0:S_p(r)\le B\}.
\]

Each slot starts its first replication and starts replication $r+1$ whenever $S_p(r)<B$. The replications with $S_p(r)\le B$ are retained. If $N_p(B)=0$, the first replication that completes beyond $B$ is retained. The second and later replications that complete beyond $B$ are excluded from the aggregate. Setting $N_p^*(B)=\max\{1,N_p(B)\}$, we define the slot mean $A_p(B)$ and the aggregate by
\[
\begin{gathered}
A_p(B)
=\frac1{N_p^*(B)}
\sum_{r=1}^{N_p^*(B)}
\widehat\theta_{p,r},\\
\widehat\theta_{\mathrm{CT}}(h_j)
=\frac1P
\sum_{p=1}^{P}A_p(B).
\end{gathered}
\]

\begin{assumption}[Parallel replication]
\label{supp:ass:parallel}
The observed data, the set of clones to be analyzed, the initial distribution $\nu$, the fixed weights $w$, and the averaging window $k:\ell$ are fixed. The logical slots $1,\ldots,P$ are fixed before the outcomes are observed. Each pair $(p,r)$ is assigned an i.i.d.\ ideal random stream, and each replication is run through to meeting with the same initial distribution and the same transition kernel. Then $(\widehat\theta_{p,r},C_{p,r})$ is i.i.d.\ within a slot and across slots. Dependence between the estimator and the computation time within the same pair is allowed.
\end{assumption}

Theorem~\ref{supp:thm:ct} states that the completion-time correction and the equal-weight average across slots preserve the expectation of the single-pair estimator.

\begin{theorem}[Equal-weight logical-slot aggregate]
\label{supp:thm:ct}
Under Assumption~\ref{supp:ass:finite} and Assumption~\ref{supp:ass:parallel}, for every clone $j$,
\[
\mathbb E[\widehat\theta_{\mathrm{CT}}(h_j)]
=\pi(h_j).
\]
\end{theorem}

\begin{proof}
Fix a clone $j$. Theorem~\ref{supp:thm:cmcpd} gives
\[
\mathbb E[\widehat\theta_{p,r}]
=\pi(h_j),
\qquad
\operatorname{Var}(\widehat\theta_{p,r})<\infty ,
\]
and the finite variance yields $\mathbb E[|\widehat\theta_{p,r}|]<\infty$. Moreover, the geometric tail of the meeting time in Theorem~\ref{supp:thm:cmcpd} together with $C_{p,r}=\max\{\ell,\tau_{p,r}\}$ gives
\[
\mathbb E[C_{p,r}]
\le \ell+\mathbb E[\tau_{p,r}]
<\infty ,
\]
and $0<C_{p,r}<\infty$ almost surely. Furthermore, since $C_{p,r}$ is a positive integer, $N_p(B)<\infty$ holds almost surely for a fixed finite $B$. By Assumption~\ref{supp:ass:parallel}, $(\widehat\theta_{p,r},C_{p,r})$ is i.i.d.\ within a slot and across slots.

Let the processor index, the replication index, the output, the runtime, the horizon, and the target mean of \citet[Sections~2--3]{glynn1991analysis} correspond, respectively, to
\[
p,\quad
r,\quad
\widehat\theta_{p,r},\quad
C_{p,r},\quad
B,\quad
\pi(h_j),
\]
and let
\[
A_p(B)
=\overline X_i\{\widetilde N_i(t)\},
\qquad
\widehat\theta_{\mathrm{CT}}(h_j)
=\widetilde\mu_1(P,t).
\]
The assumptions of \citet[Equation~(3.2), Proposition~3.2]{glynn1991analysis} therefore hold, and that proposition gives
\[
\mathbb E[\widehat\theta_{\mathrm{CT}}(h_j)]
=\pi(h_j).
\]
\end{proof}

From the same coupled pair we compute
\[
\widehat{\bm\theta}_{p,r}
=
\bigl(
H_{k:\ell}^{(p,r)}(h_1),
\ldots,
H_{k:\ell}^{(p,r)}(h_n)
\bigr)^\mathsf{T}
\]
and aggregate it as
\[
\widehat{\bm p}_{\mathrm{CT}}
=
\frac1P\sum_{p=1}^P
\frac1{N_p^*(B)}
\sum_{r=1}^{N_p^*(B)}
\widehat{\bm\theta}_{p,r}.
\]
Because all components use the common $C_{p,r}$ and $N_p(B)$, applying Theorem~\ref{supp:thm:ct} to each component gives
\[
\mathbb E[\widehat{\bm p}_{\mathrm{CT}}]
=
\bigl(\pi(h_1),\ldots,\pi(h_n)\bigr)^\mathsf{T}.
\]

\section{Details of the Experimental Setup}

The experiments used a Knill-type random pooling design with 1,298 clones, 97 pools, and 3 true positives. Each clone belongs to 3 pools. We set the prior probability to $q=2.6/1298=0.002003$, giving 2.6 unpruned prior expected positives. After naive-score pruning, 333 clones were analyzed. Setting the false-negative rate and the false-positive rate to either 0.05 or 0.10 produced four noise conditions. For Tapestry 320 and Tapestry 961, we used $q=5/320=0.015625$ and $q=10/961=0.010406$, respectively, and retained these clone-level prior probabilities after pruning. One set of observed pool data was fixed for each condition, and each cell was repeated with 10 chain seeds.

Let $n_a$ denote the number of analyzed clones and $s\in\{1,3,5,10,20\}$ the total estimator budget in sweeps. We set $M_s=s n_a$ single-site updates and $b_s=\operatorname{round}(0.1M_s)$. The uncoupled estimators discarded the first $b_s$ updates and averaged the following $M_s-b_s$ updates. The coupled estimators used $k=b_s+1$ and $\ell=M_s$, giving $\ell-k+1=M_s-b_s$ terms in the base average. Table~\ref{supp:tab:update-counts} gives the exact counts.

\begin{table*}[t]
\centering
\caption{Exact burn-in and averaging-update counts.}
\label{supp:tab:update-counts}
\small
\begin{tabular}{lll}
\toprule
Dataset & $b_s$ for $s=1,3,5,10,20$ & $M_s-b_s$ \\
\midrule
Synthetic, $n_a=333$ & 33, 100, 166, 333, 666 & 300, 899, 1499, 2997, 5994 \\
Tapestry 320, $n_a=320$ & 32, 96, 160, 320, 640 & 288, 864, 1440, 2880, 5760 \\
Tapestry 961, $n_a=247$ & 25, 74, 124, 247, 494 & 222, 667, 1111, 2223, 4446 \\
\bottomrule
\end{tabular}

\end{table*}

The uncoupled methods used 2,000,000 logical chain slots and the coupled methods 1,000,000 logical slots. Since each coupled logical slot holds two chain states, the number of chain states held simultaneously is 2,000,000 in both cases. The per-slot progress-unit horizon was fixed at 20 sweeps in every cell. For an estimator length of $s$ sweeps, the number of chains of the uncoupled methods, in units of single-site updates, is
\[
N_s
=\left\lfloor
\frac{2{,}000{,}000\times20\times333}{s\times333}
\right\rfloor
\approx\frac{40{,}000{,}000}{s}.
\]
When slot refill is enabled, a slot that has completed one pair starts the next pair within the 20-sweep horizon. The five conditions refer to estimator lengths of 1, 3, 5, 10, and 20 sweeps, and the per-slot progress-unit horizon is 20 sweeps in all conditions. The reported coupled runs also impose a 20-sweep meeting cap. They are finite-cap approximations to the uncapped estimators covered by the unbiasedness theorems.

We audited cap-hit rates at the coupled-pair level. The denominator is the total number of coupled pairs attempted in a method--condition cell, and the numerator is the number of pairs that failed to meet before the 20-sweep cap. In the maximum-noise synthetic condition, both coupled random update and coupled score-weighted update had 0\% cap-hit pairs. For coupled random update, the three remaining synthetic conditions FN $=$ FP $=0.05$, FN $=0.05$/FP $=0.10$, and FN $=0.10$/FP $=0.05$ had rates of $2.84\times10^{-3}\%$, $2.27\times10^{-4}\%$, and $1.19\times10^{-5}\%$, respectively; Tapestry 320 and Tapestry 961 had rates of 0.486\% and 10.257\%. For coupled score-weighted update, the corresponding rates were $2.44\times10^{-6}\%$, $3.04\times10^{-7}\%$, $1.52\times10^{-6}\%$, 0.00259\%, and 2.869\%.

\section{Additional Experimental Results}

The main text focused on the maximum noise condition. Here we show the top-10 comparison across the four noise conditions, MAE by estimator length, the slot refill ablation, and top-10 MAE by estimator length for the two real-data conditions.

\subsection{Reference Top-10 Comparison}

Table~\ref{supp:tab:appendix-top10-all-conditions} shows the four noise conditions with the same aggregation method as Table~1 of the main text. Within each execution the 5 estimator lengths are averaged first, and the mean and the standard error are then taken over the 10 executions.

\begin{table*}[p]
\centering
\caption{Reference top-10 comparison across four noise conditions.}
\label{supp:tab:appendix-top10-all-conditions}
\small
\begin{tabular}{lcc}
\toprule
\multicolumn{3}{c}{FN $=0.05$, FP $=0.05$} \\
\midrule
Method & Mean top-10 MAE & Mean top-10 overlap \\
\midrule
systematic scan & $(1.186{\pm}0.000){\times}10^{-2}$ & $0.828{\pm}0.003$ \\
random update & $(6.941{\pm}0.001){\times}10^{-2}$ & $\mathbf{0.920{\pm}0.000}$ \\
coupled random update & $(1.881{\pm}0.088){\times}10^{-2}$ & $0.826{\pm}0.007$ \\
score-weighted update & $(4.339{\pm}0.002){\times}10^{-3}$ & $0.860{\pm}0.000$ \\
coupled score-weighted update & $\mathbf{(2.261{\pm}0.111){\times}10^{-3}}$ & $0.862{\pm}0.007$ \\
\bottomrule
\end{tabular}

\vspace{0.8em}

\begin{tabular}{lcc}
\toprule
\multicolumn{3}{c}{FN $=0.05$, FP $=0.10$} \\
\midrule
Method & Mean top-10 MAE & Mean top-10 overlap \\
\midrule
systematic scan & $(2.487{\pm}0.000){\times}10^{-2}$ & $0.800{\pm}0.000$ \\
random update & $(2.381{\pm}0.000){\times}10^{-2}$ & $\mathbf{0.976{\pm}0.005}$ \\
coupled random update & $(1.000{\pm}0.048){\times}10^{-2}$ & $0.762{\pm}0.009$ \\
score-weighted update & $(3.364{\pm}0.001){\times}10^{-3}$ & $0.788{\pm}0.003$ \\
coupled score-weighted update & $\mathbf{(1.108{\pm}0.029){\times}10^{-3}}$ & $0.802{\pm}0.009$ \\
\bottomrule
\end{tabular}

\vspace{0.8em}

\begin{tabular}{lcc}
\toprule
\multicolumn{3}{c}{FN $=0.10$, FP $=0.05$} \\
\midrule
Method & Mean top-10 MAE & Mean top-10 overlap \\
\midrule
systematic scan & $(7.287{\pm}0.002){\times}10^{-3}$ & $0.840{\pm}0.000$ \\
random update & $(2.957{\pm}0.000){\times}10^{-2}$ & $0.930{\pm}0.003$ \\
coupled random update & $(2.669{\pm}0.166){\times}10^{-3}$ & $0.928{\pm}0.010$ \\
score-weighted update & $(4.659{\pm}0.001){\times}10^{-3}$ & $0.936{\pm}0.003$ \\
coupled score-weighted update & $\mathbf{(1.112{\pm}0.030){\times}10^{-3}}$ & $\mathbf{0.946{\pm}0.007}$ \\
\bottomrule
\end{tabular}

\vspace{0.8em}

\begin{tabular}{lcc}
\toprule
\multicolumn{3}{c}{FN $=0.10$, FP $=0.10$} \\
\midrule
Method & Mean top-10 MAE & Mean top-10 overlap \\
\midrule
systematic scan & $(1.630{\pm}0.001){\times}10^{-3}$ & $0.800{\pm}0.000$ \\
random update & $(1.350{\pm}0.000){\times}10^{-2}$ & $\mathbf{0.998{\pm}0.002}$ \\
coupled random update & $(1.248{\pm}0.093){\times}10^{-3}$ & $0.844{\pm}0.010$ \\
score-weighted update & $(3.590{\pm}0.001){\times}10^{-3}$ & $0.654{\pm}0.003$ \\
coupled score-weighted update & $\mathbf{(5.109{\pm}0.303){\times}10^{-4}}$ & $0.928{\pm}0.008$ \\
\bottomrule
\end{tabular}

\end{table*}

\subsection{MAE by Estimator Length}

Figure~\ref{supp:fig:appendix-sweep-mae-all-conditions} shows the all-clone MAE and the top-10 MAE across the four noise conditions. Each row corresponds to a noise condition; the left column refers to the 333 clones after pruning and the right column to the reference top-10 clones.

\begin{figure*}[p]
\centering
\includegraphics[width=\linewidth,height=0.86\textheight,keepaspectratio]{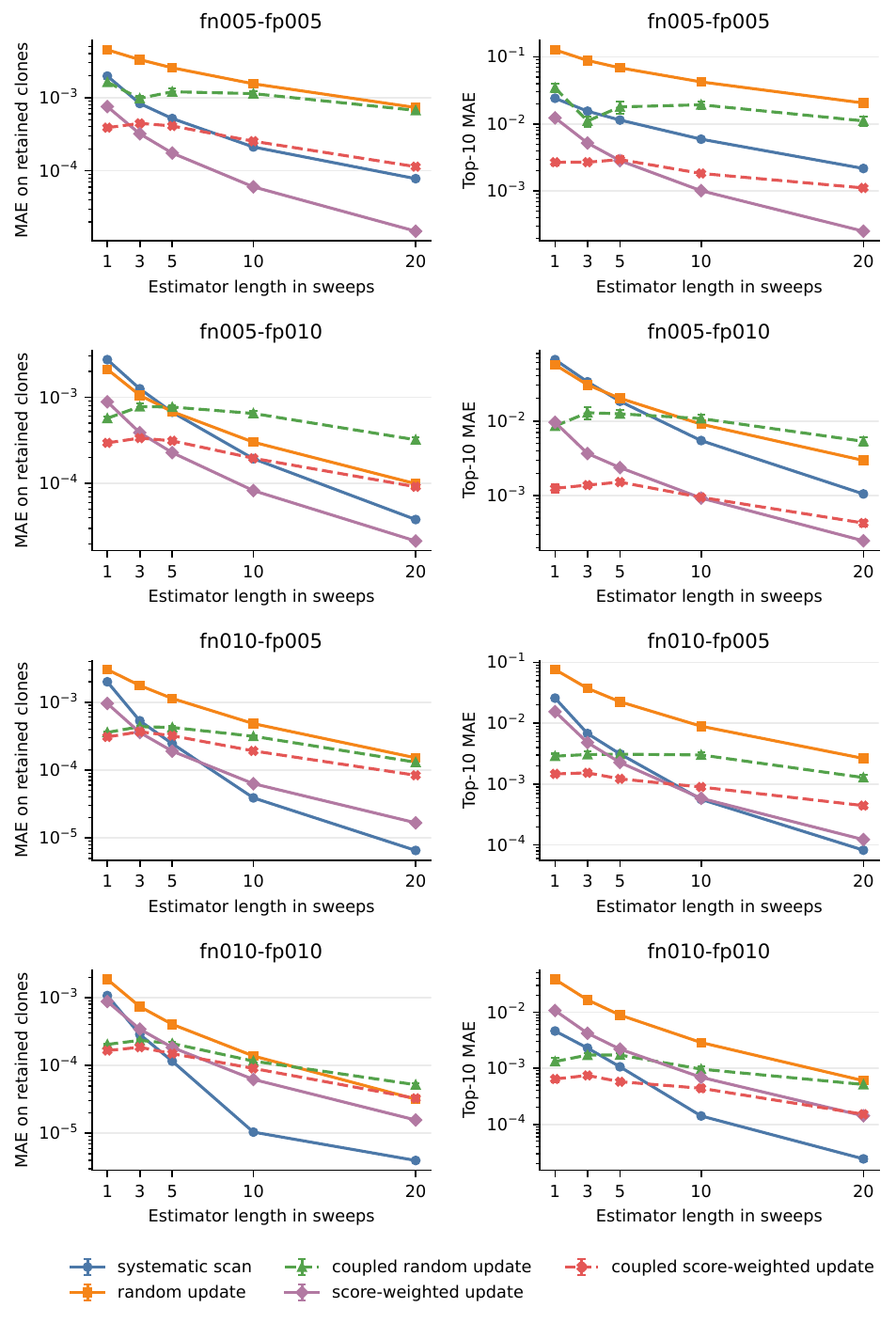}
\caption{MAE by estimator length across four noise conditions.}
\label{supp:fig:appendix-sweep-mae-all-conditions}
\end{figure*}

\subsection{Slot Refill Ablation}

Figure~\ref{supp:fig:appendix-refill-all-conditions} compares coupled score-weighted update with and without slot refill across the four noise conditions. The vertical axis is the all-clone MAE for the 333 clones after pruning.

\begin{figure*}[t]
\centering
\includegraphics[width=\linewidth]{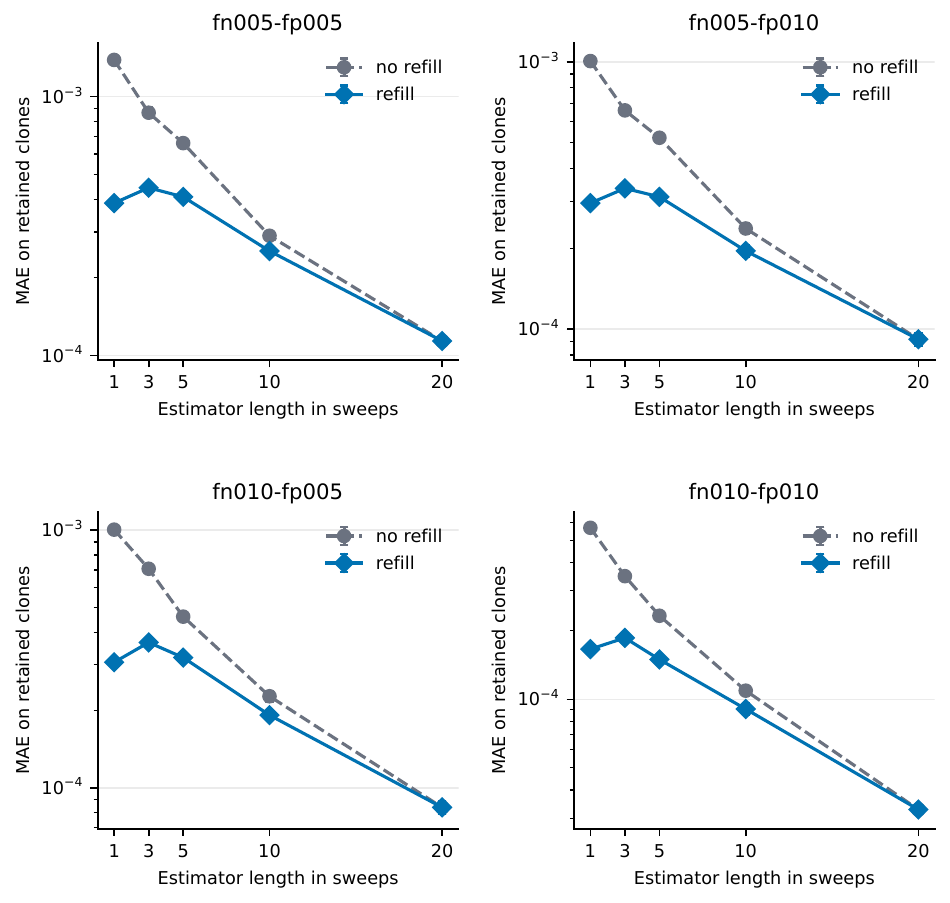}
\caption{Slot refill ablation across four noise conditions.}
\label{supp:fig:appendix-refill-all-conditions}
\end{figure*}

\subsection{MAE by Estimator Length on the Real Data}

Figure~\ref{supp:fig:appendix-real-data-sweep-mae} shows the top-10 MAE by estimator length on Tapestry 320 and Tapestry 961. Each point is the average over 10 executions, and the error bars are standard errors.

\begin{figure*}[t]
\centering
\includegraphics[width=\linewidth]{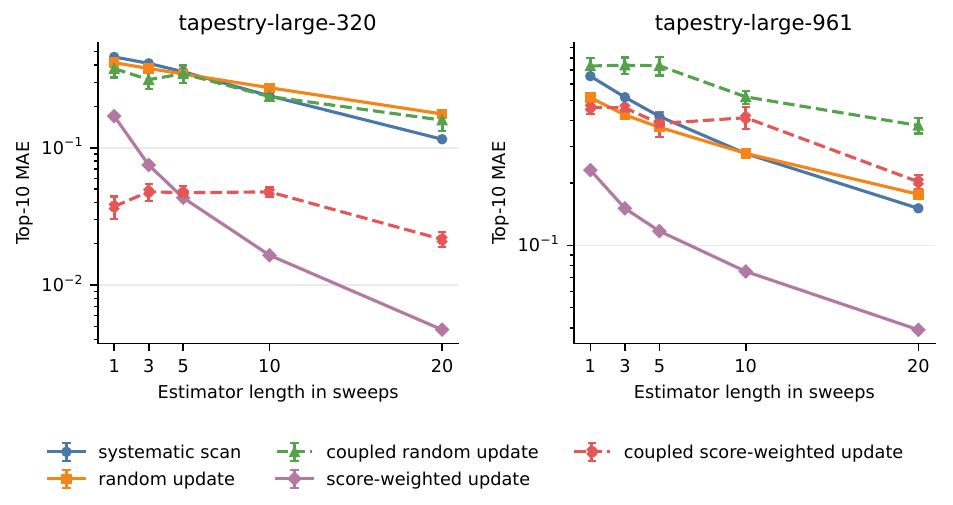}
\caption{Top-10 MAE by estimator length on Tapestry real data.}
\label{supp:fig:appendix-real-data-sweep-mae}
\end{figure*}

\section{Post Hoc Statistical Analysis}

We compared top-10 MAE within each condition and sweep using a two-sided exact independent-sample permutation test with the difference in means as the statistic. We treated methods as independent samples because their method-specific seed streams differ. For the refill ablation, the refill and no-refill variants share the same seed stream, so we used a two-sided exact paired sign-flip test on retained-clone MAE. We applied Holm correction separately to the 40 synthetic method comparisons, the 20 real-data method comparisons, and the 20 refill comparisons.

Table~\ref{supp:tab:statistical-test-summary} reports the Holm-adjusted results at level 0.05 in a single table with rows for each dataset or synthetic noise condition and estimator length. The entries lower and higher refer to the first-named method or variant in the column header, and n.s.\ denotes a non-significant result.

\begin{table*}[t]
\centering
\caption{Dataset-wise and sweep-wise summary of post hoc exact tests. Entries use Holm-adjusted $p\le0.05$; n.s.\ denotes a non-significant result.}
\label{supp:tab:statistical-test-summary}
\scriptsize
\begin{tabular}{lrccc}
\toprule
Data & Sweep & cSW vs. cRU & cSW vs. SW & refill vs. no refill \\
\midrule
FN=FP=0.05 & 1 & lower & lower & lower \\
FN=FP=0.05 & 3 & lower & lower & lower \\
FN=FP=0.05 & 5 & lower & n.s. & lower \\
FN=FP=0.05 & 10 & lower & higher & lower \\
FN=FP=0.05 & 20 & lower & higher & n.s. \\
FN=0.05/FP=0.10 & 1 & lower & lower & lower \\
FN=0.05/FP=0.10 & 3 & lower & lower & lower \\
FN=0.05/FP=0.10 & 5 & lower & lower & lower \\
FN=0.05/FP=0.10 & 10 & lower & n.s. & lower \\
FN=0.05/FP=0.10 & 20 & lower & higher & n.s. \\
FN=0.10/FP=0.05 & 1 & lower & lower & lower \\
FN=0.10/FP=0.05 & 3 & lower & lower & lower \\
FN=0.10/FP=0.05 & 5 & lower & lower & lower \\
FN=0.10/FP=0.05 & 10 & lower & higher & lower \\
FN=0.10/FP=0.05 & 20 & lower & higher & n.s. \\
FN=FP=0.10 & 1 & lower & lower & lower \\
FN=FP=0.10 & 3 & lower & lower & lower \\
FN=FP=0.10 & 5 & lower & lower & lower \\
FN=FP=0.10 & 10 & lower & lower & lower \\
FN=FP=0.10 & 20 & lower & n.s. & n.s. \\
Tapestry 320 & 1 & lower & lower & -- \\
Tapestry 320 & 3 & lower & lower & -- \\
Tapestry 320 & 5 & lower & n.s. & -- \\
Tapestry 320 & 10 & lower & higher & -- \\
Tapestry 320 & 20 & lower & higher & -- \\
Tapestry 961 & 1 & lower & higher & -- \\
Tapestry 961 & 3 & lower & higher & -- \\
Tapestry 961 & 5 & lower & higher & -- \\
Tapestry 961 & 10 & n.s. & higher & -- \\
Tapestry 961 & 20 & lower & higher & -- \\
\bottomrule
\end{tabular}

\end{table*}

These post hoc tests quantify chain-seed variation conditional on each fixed pooling design, latent state, and observed dataset. They do not provide population-level inference over pooling designs or datasets.

\end{document}